\documentclass{amsart}
\usepackage{amsfonts, amsthm}
\newtheorem{definition}{Definition}
\newtheorem{lemma}{Lemma}
\newtheorem{theorem}{Theorem}
\newtheorem{corollary}{Corollary}

\title[]{Darboux integrable discrete equations possessing an autonomous first-order integral}

\author[]{S. Ya. Startsev}
\address{Ufa Institute of Mathematics, Russian Academy of Sciences}

\begin{document}

\begin{abstract}
All Darboux integrable difference equations on the quad-graph are described in the case of the equations that possess autonomous first-order integrals in one of the characteristics. A generalization of the discrete Liouville equation is obtained from a subclass of these equation via a non-point transformation. The general proposition on the symmetry structure of the quad-graph equations is proved as an auxiliary result.
\end{abstract}

\maketitle

\section{Introduction and basic definitions}
Let us consider difference equations of the form
\[
u_{(i+1,j+1)}=F(u_{(i,j)},u_{(i+1,j)},u_{(i,j+1)}),  
\]
where $u$ is a function of two integers, the lower multi-index denotes values of the arguments for this function, and the equation holds true for any $(i,j) \in \mathbb{Z}^2$. Integrable (in various senses) equations of such form are actively studied in recent years (e.g. see \cite{R2001,ABS,LY,Mikh,LS} and references within). The present work is devoted to Darboux integrable equations, which can be considered as a special case of $C$-integrable ones. The term `$C$-integrability' was offered in \cite{Cal}, and for the Darboux integrability definition we need to introduce designations first.

From now on, we will use the notation $u_{m,n}:=u_{(i+m,j+n)}$, $u:=u_{0,0}=u_{(i,j)}$ to omit $i$ and $j$ for brevity. I.e. $u_{m,n}$ designates the function that is obtained from the function $u$ via the shifts in its first and second arguments by $n$ and $m$, respectively. According this notation, the above equation reads
\begin{equation}\label{uij}
u_{1,1}=F(u,u_{1,0},u_{0,1}).  
\end{equation}
We assume that
\begin{equation}\label{hipc}
\frac{\partial F}{\partial u} \ne 0, \qquad \frac{\partial F}{\partial u_{1,0}} \ne 0, \qquad \frac{\partial F}{\partial u_{1,0}} \ne 0.
\end{equation}
These conditions allow us to express any argument of the function $F$ in terms of the others for rewriting \eqref{uij}, after appropriate shifts in $i$ and $j$, in any of the following forms
\begin{equation}\label{umm}
u_{-1,-1}=\overline{F}(u,u_{-1,0},u_{0,-1}),  
\end{equation}
\begin{equation}\label{upm}
u_{1,-1}=\hat{F}(u,u_{1,0},u_{0,-1}),  
\end{equation}
\begin{equation}\label{ump}
u_{-1,1}=\tilde{F}(u,u_{-1,0},u_{0,1}).  
\end{equation}
Using \eqref{uij}--\eqref{ump} and their consequences derived by shifts in $i$ and $j$, we can express any `mixed shift' $u_{m,n}$ (for both positive and negative non-zero $n$ and $m$) in terms of the functions $u_{k,0}$, $u_{0,l}$. Thus, we can (and will) formulate reasonings of this article only in terms of an arbitrary solution $u$ to \eqref{uij} and its `canonical shifts' $u_{k,0}$, $u_{0,l}$, , $k,l \in \mathbb{Z}$, which are called {\it dynamical variables}. (A more detailed explanation of the dynamical variables and the recursive procedure of the mixed shift elimination can be found, for example, in \cite{Mikh}.) 

Values of the dynamical variables for fixed $i$ and $j$ serves as boundary conditions of the Goursat problem for the equation \eqref{uij} and can be selected in an arbitrary way. Therefore, we treat the dynamical variables as functionally independent in the relationships that are valid for any solution of the equation \eqref{uij}, i.e. in all relationships stated below. The notation $g[u]$ means that the function $g$ depends on a finite number of the dynamical variables. The considerations in this paper are local (for example, we use the local implicit function theorem to obtain \eqref{umm}--\eqref{ump}) and, for simplicity, all functions are assumed to be locally analytical.

Now we let $T_{\rm i}$ and $T_{\rm j}$ denote the operators of the forward shifts in $i$ and $j$ by virtue of the equation~\eqref{uij}. The inverse (backward) shift operators are denoted by $T_{\rm i}^{-1}$ and $T_{\rm j}^{-1}$. We use a shift operator with a superscript $k$ to designate the $k$-fold application of this operator (e.g. $T_{\rm j}^3 := T_{\rm j} \circ T_{\rm j} \circ T_{\rm j}$, $T_{\rm i}^{-2} := T_{\rm i}^{-1} \circ T_{\rm i}^{-1}$ and $T_{\rm i}^1 := T_{\rm i}$). For a more compact notation we also set any operator with the zero superscript equal to the operator of multiplication by unit (i.e. the identity mapping). In these notations the shift operators are defined by the following rules:
\[
T_{\rm i}^k(f(a,b,c,\dots))=f(T_{\rm i}^k(a),T_{\rm i}^k(b),T_{\rm i}^k(c),\dots),
\]
\[
T_{\rm j}^k (f(a,b,c,\dots))=f(T_{\rm j}^k (a),T_{\rm j}^k (b),T_{\rm j}^k (c),\dots),
\]
\[
T_{\rm i}^k(u_{m,0})=u_{m+k,0}, \qquad T_{\rm j}^k (u_{0,m})=u_{0,m+k}
\]
for any function $f$ and any integers $k$ and $m$; for any $n \in \mathbb{N}$ the relationships 
\[
T_{\rm i}(u_{0,n})=T_{\rm j}^{n-1}(F), \qquad T_{\rm i}(u_{0,-n})=T_{\rm j}^{1-n}(\hat{F}),
\]
\[
T_{\rm j}(u_{n,0})=T_{\rm i}^{n-1}(F), \qquad T_{\rm j}(u_{-n,0})=T_{\rm i}^{1-n}(\tilde{F}),
\]
\[
T_{\rm i}^{-1}(u_{0,n})=T_{\rm j}^{n-1}(\tilde{F}), \qquad T_{\rm i}^{-1}(u_{0,-n})=T_{\rm j}^{1-n}(\overline{F}),
\]
\[
T_{\rm j}^{-1}(u_{n,0})=T_{\rm i}^{n-1}(\hat{F}), \qquad T_{\rm j}^{-1}(u_{0,-n})=T_{\rm i}^{1-n}(\overline{F})
\]
hold true (i.e. mixed variables $u_{1,\pm n}$, $u_{\pm n,1}$, $u_{-1,\pm n}$ and $u_{\pm n,-1}$ are expressed in terms of the dynamical variables by virtue of~\eqref{uij}-\eqref{ump}).

\begin{definition} An equation of the form \eqref{uij} is called Darboux integrable if there exist functions $I[u]$ and $J[u]$ such that the relations  $T_{\rm j}(I)=I$ and $T_{\rm i}(J)=J$ hold true and each of the functions essentially depends on at least one of the dynamical variables. In this case, the functions $I[u]$ and $J[u]$ are respectively called an $i$-integral and a $j$-integral of the equation~\eqref{uij}.
\end{definition}
It is easy to check (see, for example, \cite{umj} or lemma~\ref{l1} below) that $i$-integrals can not depend on the dynamical variables of the form $u_{0,p}$, and $j$-integrals -- on the dynamical variables of the form $u_{q,0}$. Thus, $i$- and $j$-integrals have the form $I(u_{k,0},u_{k+1,0},\dots,u_{m,0})$ and $J(u_{0,l},u_{0,l+1},\dots,u_{0,n})$, respectively.  The numbers $m-k$ and $n-l$ are called {\it order} of the corresponding integral. It should be noted that the present paper (in contrast to, for instance, \cite{HZS,GY}) deals only with the autonomous integrals, i.e. with the integrals that do not  depend explicitly on the discrete variables $i$ and $j$.

The simplest example of a Darboux integrable equation is 
\[ u_{1,1}=u_{1,0}+u_{0,1}-u. \]
This equation has the integrals $I[u]=u_{1,0}-u$, $J[u]=u_{0,1}-u$ and can be considered as a difference analogue of the partial differential wave equation $u_{xy}=0$. Another example is the equation
\begin{equation}\label{dle}
u_{1,1}=\frac{(u_{1,0}-1)(u_{0,1}-1)}{u} 
\end{equation}
from \cite{Hirota}. It is similar to the well-known Liouville equation $u_{xy}={\rm e}^u$ in properties and, according to \cite{AdS}, has the second-order integrals
\[
\quad I[u]=\left( \frac{u_{2,0}}{u_{1,0}-1} +1 \right) \left( \frac{u-1}{u_{1,0}} + 1 \right),\]
\[ J[u]=\left( \frac{u_{0,2}}{u_{0,1}-1} +1 \right) \left( \frac{u-1}{u_{0,1}} + 1 \right).
\] 

In general, the equations \eqref{uij} can be regarded as difference analogues of the partial differential equations
\begin{equation}
u_{xy}=F(u,u_x,u_y). \label{hyp}
\end{equation}
The concept of the Darboux integrability was initially introduced for partial differential equations in classical works such as \cite{gour}, and the complete classification of Darboux integrable equations of the form \eqref{hyp} was performed in \cite{ZhSok}. At present, a classification is absent for Darboux integrable equations \eqref{uij} and only separate examples of such equations are known (see, for instance, \cite{GY,Sit,HZS2}). This is why a classification problem for a special case of the equations \eqref{uij} looks reasonable and may be a natural part of the future complete classification. We consider the equations of the form
\begin{equation}\label{ji}
\phi(u_{1,0},u_{1,1})= \phi(u,u_{0,1}), \qquad \frac{\partial \phi(w,z)}{\partial w} \frac{\partial \phi(w,z)}{\partial z} \ne 0, 
\end{equation}
as such a special case. It is obvious that an equation of the form \eqref{uij} possesses a $j$-integral $\phi(u,u_{0,1})$ and satisfies conditions \eqref{hipc} if and only if it can be written as \eqref{ji}. Thus, the classification of Darboux integrable equations \eqref{ji} is reduced to finding necessary and sufficient conditions of the $i$-integral existence. These conditions are found in section~\ref{s3} and has a constructive form (i.e. allows to straightforwardly obtain any Darboux integrable equation \eqref{ji}).

To prove this main result, we need some auxiliary statements. One of them, the proposition on splitting symmetries of equation \eqref{uij} into summands involving either shifts only in $i$ or shifts only in $j$, is useful not only in the context of the present article and known to specialists but, to the best author's knowledge, has no published proof for the general form of this statement. This proof is given in section~\ref{ss1}.

A part of equations \eqref{ji} can be rewritten in the form $T_{\rm j}(a(u,u_{1,0}))=b(u,u_{1,0})$, where the functions $a$ and $b$ are functionally independent, and hence admits the non-point invertible transformation $v=a(u,u_{1,0})$ (see \cite{Sit} for more details). In section~\ref{s3} we show that such equations exist among Darboux integrable equations \eqref{ji} too and the transformation $v=a(u,u_{1,0})$ maps they into a family of Darboux integrable equations, which can be considered as a generalization of the equation \eqref{dle}. It can be proved that this family contains all Darboux integrable equations of the form $T_{\rm i}(\Omega(u,u_{0,1}))=\Psi(u,u_{0,1})$ possessing second-order autonomous $j$-integrals but this is beyond the scope of the present paper. 

We should emphasize that equations \eqref{uij} do not necessarily have the form \eqref{ji} if they possess first-order $j$-integrals depending explicitly on $i$ (examples of such kind can be found in \cite{GY}). Therefore, the main result of the present paper does not give an exhaustive description of all Darboux integrable equations \eqref{uij} possessing first-order $j$-integrals.

\section{Auxiliary statements}

Before going on, we need to define a new term.
\begin{definition} An equation $u_{t} = f[u]$ is called a symmetry of equation \eqref{uij} if the relation $L(f)=0$ holds true, where
\begin{equation}\label{lop}
L = T_{\rm i} T_{\rm j} - \frac{\partial F}{\partial u_{1,0}} T_{\rm i} - \frac{\partial F}{\partial u_{0,1}} T_{\rm j} -
\frac{\partial F}{\partial u}.
\end{equation}
\end{definition}
The sketch of further reasonings can be summarized as follows. According to \cite{AdS}, if an equation of the form \eqref{uij} is Darboux integrable and possesses a $j$-integral $\phi$, then there exists an operator $R=\sum_{q=0}^{r} \lambda_q [u] T_{\rm j}^q$, $\lambda_r \ne 0$, such that
\begin{equation}\label{ed}
u_t=R(\xi(T_{\rm j}^{p}(\phi), T_{\rm j}^{p+1}(\phi), T_{\rm j}^{p+2}(\phi), \dots))
\end{equation}
is a symmetry of this equation for any integer $p$ and any function $\xi$ depending on a finite number of the arguments. In section~\ref{ss2} we show that any symmetry of the equation \eqref{uij} is mapped into an equation of the form $v_t=g[v]$ by the substitution $v=\phi [u]$ if $\phi [u]$ is the integral of smallest order. On the other hand, the work \cite{foi} completely describes the substitutions of the form $v=\phi (u,u_{0,1})$ for the equations \eqref{ed} with the right-hand side independent of variables $u_{l,0}$, $l \in \mathbb{Z}\backslash\{0\}$. This gives us a necessary condition of Darboux integrability for the equations \eqref{ji} if we show that their symmetries \eqref{ed} do not depend on variables of the form $u_{l,0}$. The latter is done in section~\ref{ss1} by proving a general proposition on symmetry structure.

\subsection{Symmetry structure}\label{ss1}
\begin{theorem}\label{struc} Any symmetry  $u_{t} = f[u]$ of equation \eqref{uij} has the form
\begin{equation}\label{ras}
u_{t} = \hat{f} (u_{m,0}, u_{m+1,0}, u_{m+2,0}, \dots ) + \bar{f} (u_{0,n}, u_{0,n+1}, u_{0,n+2}, \dots ),
\end{equation}
i.e.  $f[u]$ is the sum of two terms such that the first one does not depend on $u_{0,k}$ and the second one does not depend on $u_{k,0}$ for any non-zero $k \in \mathbb{Z}$.
\end{theorem}
The statement of the theorem is not new. Specialists in symmetries of discrete equations believe it is true because the analogous proposition for the equations \eqref{hyp} is well-known and proved in \cite{avz}. But the proof of theorem~\ref{struc} is available, for example, in \cite{RH} for only a simple special case of symmetries depending on five dynamical variables and, to the author's best knowledge, has been absent 
for symmetries of the general form and an arbitrary high order. It is convenient to give the general proof by using the following simple proposition.
\begin{lemma}\label{l1} Let a function $g[u]$ satisfy a relationship of the form
\begin{equation}\label{lrel}
A[u] T_{\rm j} (g) + B[u] g = C[u],\qquad A B \ne 0,
\end{equation}
where $A$, $B$ and $C$ do not depend on $u_{0,k}$ for any non-zero integer $k \ne 1$. Then $g[u]$ does not depend on $u_{0,k}$ for any non-zero $k \in \mathbb{Z}$.                             
\end{lemma}
From now on, we will use the notation 
\[ f'_{p:} := \frac{\partial f}{\partial u_{p,0}},\qquad f'_{:q} := \frac{\partial f}{\partial u_{0,q}},\qquad f''_{p:q} := \frac{\partial^2 f}{\partial u_{0,q} \partial u_{p,0}}. \]
to denote the partial derivatives of a function $f[u]$ in in-line formulas. 
\begin{proof}
Assume the contrary. Let $l$ and $s$ be respectively the largest positive and the smallest negative integers for which $g[u]$ depends on $u_{0,l}$ and $u_{0,s}$. Differentiation of \eqref{lrel} with respect to $u_{0,l+1}$
and $u_{0,s}$ gives rise to $T_{\rm j}(g'_{:l})=0$ and $g'_{:s}=0$, respectively. Thus, we arrive to a contradiction that proves the lemma.
\end{proof}
\begin{proof}[Proof of theorem~\ref{struc}] It is obvious that the symmetry $u_t=f[u]$ has the form \eqref{ras} if and only if $f''_{p:q} = 0$ for all non-zero integers $p$ and $q$. Assume the contrary. Let $l$ and $s$ be respectively the largest and the smallest non-zero integers for which there exist non-zero integers $\delta$ and $\sigma$ such that $f''_{l:\delta} \ne 0$ and $f''_{s:\sigma} \ne 0$. Then $f=\hat{f}[u]+\bar{f}[u]$, where $\hat{f}'_{:r}=0$ for all non-zero integers $r$ and $\bar{f}'_{k:}=0$ for all non-zero integers $k \notin [s,l]$.

If $l>0$, then the differentiation of the relationship $L(\hat{f}+\bar{f})=0$ with respect to $u_{l+1,0}$ gives us
\[
\frac{\partial T_{\rm i}^{l}(F)}{\partial u_{l+1,0}}\, T_{\rm j}(g) - \frac{\partial F}{\partial u_{1,0}}\, g + \frac{\partial L(\hat{f})}{\partial u_{l+1,0}}=0,  
\]
where $g=T_{\rm i}(\bar{f}'_{l:})$. According to lemma~\ref{l1}, the functions $g$ and $\bar{f}'_{l:}=T_{\rm i}^{-1}(g)$ do not depend on $u_{0,q}$ for any non-zero integer $q$. And this contradicts the assumption $\bar{f}''_{l:\delta} \ne 0$.

If $l<0$, then $s<0$ too. Differentiating $L(\hat{f}+\bar{f})=0$ with respect to $u_{s,0}$, we obtain  
\[
\frac{\partial F}{\partial u_{0,1}}\, \frac{\partial T_{\rm i}^{s+1}(\tilde{F})}{\partial u_{s,0}}\, T_{\rm j}(g) + \frac{\partial F}{\partial u}\, g - \frac{\partial L(\hat{f})}{\partial u_{s,0}}=0,  
\]
where $g=\bar{f}'_{s:}$ and $\tilde{F}$ is the right-hand side of \eqref{ump}. Hence, the assumption $\bar{f}''_{s:\sigma} \ne 0$ contradicts lemma~\ref{l1}.
\end{proof}

\begin{corollary}\label{c1}
Let an equation of the form \eqref{uij} possess a $j$-integral $\phi$ and there exists an operator $R=\sum_{k=0}^{r} \lambda_k [u] T_{\rm j}^k$ such that \eqref{ed} is a symmetry of this equation for any integer $p$ and any function $\xi$ depending on a finite number of the arguments. Then the coefficients $\lambda_k$ of the operator $R$ do not depend on $u_{q,0}$ for any non-zero integer $q$.
\end{corollary}
\begin{proof} Let $n$ be the largest integer for which $\phi'_{:n} \ne 0$. Then we set $\xi = T_{\rm j}^p(\phi)$, where $p$ is selected so that all $\lambda _k$ do not depend on $u_{0,m}$ for any $m \ge n+p$.

Now assume the contrary again. Let $l$ be the largest number for which there exists a non-zero integer $s$ such that $\lambda_l$ depend on $u_{s,0}$. Then
\[ \frac{\partial R(T_{\rm j}^p(\phi))}{\partial u_{s,0}} = \sum_{k=0}^{l} \frac{\partial \lambda _k}{\partial u_{s,0}} T_{\rm j}^{p+k}(\phi),\qquad \frac{\partial ^2 R(T_{\rm j}^p(\phi))}{ \partial u_{s,0}\,\partial u_{0,n+p+l}} = \frac{\partial \lambda _k}{\partial u_{s,0}} \frac{\partial T_{\rm j}^{p+l}(\phi)}{\partial u_{0,n+p+l}} \ne 0.\]
The last inequality contradicts theorem~\ref{struc}.
\end{proof}

\subsection{Integrals as substitutions}\label{ss2}

Let us consider the chain of differential-difference equations
\begin{equation}\label{eg}
u_t=g(u_{0,k},u_{0,k+1},\dots,u_{0,n}).
\end{equation}
It should be noted that any equation of the form $u_t=f[u]$ generates the differentiation $\partial_f$ with respect to $t$ by virtue of this equation. On the functions of the dynamical variables, the differentiation $\partial_f$ is defined by the formula ${\partial_f (h[u]) = h_* (f)}$, where
\[ h_* = \sum_{q=-\infty}^{+\infty} \frac{\partial h}{\partial u_{q,0}} T^q_{\rm i} + \sum_{q=-\infty \atop q \ne 0}^{+\infty} \frac{\partial h}{\partial u_{0,q}} T^q_{\rm j}, \]
i.e. $h_*$ is the linearization operator (Frech\'et derivative) of $h$.

\begin{definition}\label{drp}
We say that equation \eqref{eg} admits a difference substitution
\begin{equation}\label{rp}
v=\phi (u_{0,l},u_{0,l+1},\dots,u_{0,m})
\end{equation}
into an equation of the form $v_t=\hat{g}(v_{0,k},v_{0,k+1},\dots,v_{0,n})$ if the function $\phi$ depends on at least two dynamical variables and the relation
\begin{equation}\label{dpd}
\partial _g (\phi) = \hat{g} (T^k_{\rm j}(\phi), T^{k+1}_{\rm j}(\phi), \dots, T^n_{\rm j}(\phi))
\end{equation}
holds true (i.e. $v$ is a solution of the equation $v_t=\hat{g}$ for any solution of \eqref{eg}).

We call \eqref{rp} a Miura-type substitution if there exist operators
\begin{equation}\label{rop}
R=\sum_{q=0}^r \lambda_q (u_{0,\varrho},u_{0,\varrho +1},\dots,u_{0,s}) T^q_{\rm j},\qquad \lambda_r \ne 0,  
\end{equation}
\[ \hat{R}=\sum_{q=l}^{r+m} \hat{\lambda}_q (v_{0,\hat{\varrho}},v_{0,\hat{\varrho}+1},\dots,v_{0,\hat{s}}) T^q_{\rm j} \]
such that the equation $u_t=R(\xi (T^{p}_{\rm j} (\phi), T^{p+1}_{\rm j} (\phi), \dots))$ admits the substitution \eqref{rp} into the equation $v_t=\hat{R} (\xi (v_{0,p},v_{0,p+1},\dots))$ for any integer $p$ and any function $\xi$ depending on a finite number of the arguments.
\end{definition}

It is easy to see that the above definition in no way uses equation \eqref{uij} because the shift operator $T_{\rm j}$ is applied here to the functions depending only on the variables of the form $u_{0,p}$. However, the integrals of equations \eqref{uij} can be interpreted as substitutions for equations \eqref{eg}. To show this, we use the two lemmas below, which allow us to transfer appropriate reasoning for equations \eqref{hyp} from the work \cite{ZhSok} to the case of difference equations with almost no changes.

\begin{lemma}\label{l2} Let $\phi[u]$ be a $j$-integral of the smallest order for an equation of the form \eqref{uij}. Then for any $j$-integral $J[u]$ of this equation there exists a function $\xi$ such that
\begin{equation}\label{sti}
J[u] = \xi (T_{\rm j}^p(\phi), T_{\rm j}^{p+1}(\phi), T_{\rm j}^{p+2}(\phi), \dots, T_{\rm j}^{q}(\phi))
\end{equation}
for some integers $p$ and $q$.
\end{lemma}
The proof of lemma~\ref{l2} is omitted because it coincides with the proof of the analogous proposition for semi-discrete equations in \cite{HZS} (see theorem~3.2 within). 
It is obvious that the converse is also true: the right-hand side of \eqref{sti} is a $j$-integral for any $\xi$, $p$ and $q$ because $T_{\rm j}^s$ commutes with $T_{\rm i}$ for any $s \in \mathbb{Z}$ and therefore maps $j$-integrals into $j$-integrals again.

\begin{lemma}\label{l3} For any function $h[u]$ there exists an operator $H= \sum_{s=p}^{q} \mu_s[u] T^s_{\rm j}$ such that
\begin{equation}\label{tlin}
(T_{\rm i} (h))_* -  T_{\rm i} \circ h_* = H \circ L,
\end{equation}
where $\circ$ denotes the composition of operators and $L$ is defined by \eqref{lop}.
\end{lemma}
In particular, this lemma implies that the differentiation $\partial_f$ commutes with $T_{\rm i}$ if $u_t=f[u]$ is a symmetry of equation \eqref{uij}.
\begin{proof}
It is easy to check that $T_{\rm j} \circ \eta[u]_* - T_{\rm j} (\eta[u])_* = T_{\rm j}(\eta'_{1:}) L$ if $\eta'_{k:} =0$ for all non-zero integers $k \ne 1$. This implies 
\begin{equation}\label{tip}
T_{\rm j}^n (F)_* = T_{\rm j}^n \circ F_* - \sum_{r=1}^{n} T_{\rm j}^{r} \left( \frac{\partial T_{\rm j}^{n-r} (F)}{\partial u_{1,0}}\right) T_{\rm j}^{r-1} \circ L, \quad \forall n>0.
\end{equation}
Differentiating the consequence $T_{\rm j}^{-1}(F)=F(u_{0,-1},\hat{F},u)=u_{1,0}$ of \eqref{uij} with respect to $u$, $u_{0,-1}$ and $u_{1,0}$, and denoting $\theta :=-1/F'_{1:}$, we obtain
\[\frac{\partial \hat{F}}{\partial u}= T_{\rm j}^{-1} \left( \theta \frac{\partial F}{\partial u_{0,1}}\right),\qquad \frac{\partial \hat{F}}{\partial u_{0,-1}}= T_{\rm j}^{-1} \left( \theta \frac{\partial F}{\partial u}\right),\qquad \frac{\partial \hat{F}}{\partial u_{1,0}} = - T_{\rm j}^{-1} (\theta),\]
and hence $T_{\rm j}^{-1} T_{\rm i} - \hat{F}_* = T_{\rm j}^{-1} \circ \theta L$. This gives us
$T_{\rm j}^{-1} \circ \eta[u]_* - T_{\rm j}^{-1} (\eta[u])_* = T_{\rm j}^{-1} \circ \theta \eta'_{1:} L$  
if $\eta'_{k:} =0$ for all non-zero integers $k \ne 1$, and
\begin{equation}\label{tim}
T_{\rm j}^n (\hat{F})_* = T_{\rm j}^n \circ \hat{F}_* - \sum_{r=n}^{-1} T_{\rm j}^{r} \left( \theta \frac{\partial T_{\rm j}^{n-r} (\hat{F})}{\partial u_{1,0}}\right) T_{\rm j}^r \circ L, \quad \forall n<0.
\end{equation}
At the same time, the left-hand side of \eqref{tlin} is equal to
\[ \sum_{s=1}^{s=+\infty} \left( T_{\rm i} \left(\frac{\partial h}{\partial u_{0,s}}\right)( T_{\rm j}^{s-1}(F)_* - T_{\rm i} T_{\rm j}^s) + T_{\rm i} \left(\frac{\partial h}{\partial u_{0,-s}}\right)( T_{\rm j}^{1-s}(\hat{F})_* - T_{\rm i} T_{\rm j}^{-s}) \right) \]
\begin{eqnarray*}
= \sum_{s=1}^{s=+\infty} &\left( T_{\rm i} \left(\frac{\partial h}{\partial u_{0,s}}\right)( T_{\rm j}^{s-1}(F)_* - T_{\rm j}^{s-1}\circ F_*  - T_{\rm j}^{s-1} \circ L) \right. \\
& \left. + T_{\rm i} \left(\frac{\partial h}{\partial u_{0,-s}}\right)( T_{\rm j}^{1-s}(\hat{F})_* - T_{\rm j}^{1-s} \circ \hat{F}_* - T_{\rm j}^{-s} \circ \theta L) \right).
\end{eqnarray*}
Taking \eqref{tip} and \eqref{tim} into account, we therefore obtain \eqref{tlin}.
\end{proof}

\begin{theorem}\label{teo2}
Assume that equation \eqref{uij} possesses $j$-integrals and a symmetry of the form \eqref{eg}. Let $\phi[u]$ be a $j$-integral of the smallest order for this equation. Then the equation \eqref{eg} admits the substitution $v=\phi[u]$ into an equation of the form $v_t=\hat{g}(v_{0,k},v_{0,k+1},\dots,v_{0,n})$. If, in addition, the equation \eqref{uij} is Darboux integrable, then $v=\phi[u]$ is a Miura-type substitution.
\end{theorem}
\begin{proof} According to lemma~\ref{l3}, if $T_{\rm i} (\phi)= \phi$ and $L(g)=0$, then $T_{\rm i} (\partial_g(\phi))= \partial_g(\phi)$. Thus, $\partial_g(\phi)$ is a $j$-integral, and lemma~\ref{l2} implies that \eqref{dpd} holds true for some function $\hat{g}$ (it is easy to check that $p$ and $q$ in \eqref{sti} must coincide with $k$ and $n$ if $J=\partial_g(\phi)$).

If the equation \eqref{uij} is Darboux integrable, then it possesses symmetries \eqref{ed}. The operator $R$ in \eqref{ed} has the form \eqref{rop} by corollary~\ref{c1}. Thus, $v=\phi[u]$ is a Miura-type substitution.  
\end{proof}

\section{The classification result and examples}\label{s3}
\begin{theorem}Equation \eqref{ji} is Darboux integrable if and only if for $\phi(u,u_{0,1})$ there exist functions  $\alpha$, $\beta$, $\gamma$ and $\zeta$ such that
\begin{equation}\label{iii}
\zeta(u_{0,1})=\alpha (\phi) + \frac{\beta(\phi)}{\gamma(\phi)-\zeta(u)}, \quad \beta \zeta' \ne 0.
\end{equation}
\end{theorem}
It is obvious that \eqref{iii} can hold true only if $|\alpha'|+|\beta'|+|\gamma'| \ne 0$. 
\begin{proof} If equation \eqref{ji} is Darboux integrable, then $v = \phi (u,u_{0,1})$ is a Miura-type substitution by theorem~\ref{teo2}. But it is proved in \cite{foi} that $v = \phi (u,u_{0,1})$ is a Miura-type substitution if and only if $\phi$ satisfies a relationship of the form \eqref{iii}.

Conversely, let $\phi$ satisfy a relationship of the form \eqref{iii}. Then the equation \eqref{ji} is Darboux integrable because
\begin{equation}\label{gin} 
I = \frac{\zeta(u_{3,0})-\zeta(u_{1,0})}{\zeta(u_{3,0})-\zeta(u_{2,0})} \cdot \frac{\zeta(u_{2,0})-\zeta(u)}{\zeta(u_{1,0})-\zeta(u)} 
\end{equation}
is an $i$-integral of this equation. Indeed, \eqref{iii} and $T_{\rm i} (\phi) = \phi$ imply that
\[ T_{\rm j} (\zeta(u_{k,0}))= T_{\rm i}^k (\zeta(u_{0,1})) = \alpha (\phi) + \frac{\beta(\phi)}{\gamma(\phi)-\zeta(u_{k,0})} \]
for any integer $k$. Taking this into account, we obtain 
\[ T_{\rm j} (\zeta(u_{k,0}) - \zeta(u_{n,0}))= \beta(\phi) \frac{\zeta(u_{k,0}) - \zeta(u_{n,0})}{(\gamma(\phi)-\zeta(u_{k,0}))(\gamma(\phi)-\zeta(u_{n,0}))} \]
and $T_{\rm j}(I)=I$.
\end{proof}

The above theorem means that any Darboux integrable equation \eqref{uij} possessing an autonomous first-order $j$-integral is related via a point transformation $\tilde{u}=\zeta(u)$ to an equation of the form 
\begin{equation*}
\tilde{u}_{1,1}=\alpha (\varphi) + \frac{\beta(\varphi)}{\gamma(\varphi)-\tilde{u}_{1,0}},\qquad \beta \ne 0, 
\end{equation*}
where $\varphi(\tilde{u},\tilde{u}_{0,1})$ is defined by the relationship
\begin{equation*}
\tilde{u}_{0,1}=\alpha (\varphi) + \frac{\beta(\varphi)}{\gamma(\varphi)-\tilde{u}},
\end{equation*}
and $\phi =\varphi(\zeta(u),\zeta(u_{0,1}))$ is the $j$-integral of the equation \eqref{uij}.

Although we can assume without loss of generality that $\zeta(u)=u$, another choice of $\zeta$ is sometimes convenient to obtain a more simple form of the equation. For example, the equation
\begin{equation}\label{ex1}
\phi=(u+1)/u_{0,1}, \qquad \frac{u_{1,1}}{u_{0,1}} =\frac{u_{1,0}+1}{u+1}
\end{equation} 
corresponds to the choice $\zeta(u)=u^{-1}$, $\alpha(\phi)=\beta(\phi)=\phi$ and $\gamma=-1$, and takes the slightly less memorable form if $\zeta(u)=u$ is used with the same $\alpha$, $\beta$ and $\gamma$. This equation was obtained in \cite{GY} and has the $i$-integral $(u_{2,0}-u_{1,0})/(u_{1,0}-u)$. The latter illustrates that some Darboux integrable  equations \eqref{ji} can possess $i$-integrals of the order less than 3, while \eqref{gin} guarantees the existence of the third order $i$-integral only.

Up to the point transformation, the above example is a particular case of the equation that is generated by choice $\zeta(u)=u$, $\alpha(\phi)=\delta \phi$, $\beta(\phi)=D - C \phi -\delta \phi (A \phi - B)$ and $\gamma=A \phi - B$, where $A$, $B$, $C$, $D$ and $\delta$ are arbitrary constants such that $|\delta|+|A|+|C| \ne 0$, $|\delta A| + |C - \delta B| + |D| \ne 0$. Without loss of generality, we can assume that $\delta$ equals $1$ or $0$. Substituting this choice into \eqref{iii} and solving it for $\phi$, we obtain
\[ \phi = \frac{u_{0,1}(u + B)+D}{A u_{0,1} + \delta u+C}. \]
The corresponding equation \eqref{ji} after solving for $u_{1,1}$ takes the form
\[ u_{1,1}=\frac{((u+B)(\delta u_{1,0} +C) - AD)u_{0,1} + \delta D (u_{1,0}- u)}{A(u_{1,0}-u) u_{0,1} + (u_{1,0}+B)(\delta u +C) - AD} \]
and can be rewritten as
\begin{equation}\label{strt}
T_{\rm j} \left( \frac{\delta D - u_{1,0} (A u + C - \delta B)}{u_{1,0}-u} \right) =  \frac {(u+B)(\delta u_{1,0} +C) - AD}{u_{1,0}-u}.
\end{equation}

The latter formula means that this equation admits the non-point invertible transformation (see \cite{Sit} for more details). This allows us to derive another Darboux integrable equation from \eqref{strt} (and this is why such $\alpha$, $\beta$ and $\gamma$ are chosen). Let 
\begin{equation}\label{u10}
v=\frac{\delta D - u_{1,0} (A u + C - \delta B)}{u_{1,0}-u}. 
\end{equation}
Then \eqref{strt} implies that
\begin{equation}\label{u11}
v_{0,1}:=T_{\rm j}(v) = \frac {(u+B)(\delta u_{1,0} +C) - AD}{u_{1,0}-u}.
\end{equation} 
Expressing $u_{1,0}$ in the term of $u$ and $v$ from \eqref{u10}, we obtain
\begin{equation}\label{uint}
u_{1,0}=\frac{u v + \delta D}{v + A u + C - \delta B}.
\end{equation} 
The substitution of this expression into \eqref{u11} gives rise to the quadratic equation  
$P_2 u^2 + P_1 u + P_0 =0$ with the coefficients
\[ P_2 = \delta v + A v_{0,1} + A C, \]
\[ P_1 = (C + \delta B) v + (C - \delta B) v_{0,1} + (B C - A D - \delta D) (A - \delta) + C^2, \]
\[ P_0 =  (B C - A D) v - \delta D v_{0,1} + B C (C - \delta B) + D (\delta A B - A C + \delta^2 B). \]
Thus, the transformation \eqref{u10} maps solutions of \eqref{strt} into solutions of the equation
\begin{equation}\label{ggen}
T_{\rm i}(\theta (v,v_{0,1}))=\frac{v \theta (v,v_{0,1}) + \delta D}{A \theta(v,v_{0,1}) + v + C - \delta B},
\end{equation}
where $\theta$ is a root of the polynomial $P_2 \theta^2 + P_1 \theta + P_0$. Repeating the above reasonings in the inverse order, it is easy to see that the transformation $u=\theta (v,v_{0,1})$ maps solutions of \eqref{ggen} back into solutions of \eqref{strt}. This implies that rewriting the integrals of \eqref{strt} in terms of $v$ and its shifts gives us integrals of \eqref{ggen}. Therefore,
\[ J[v]=\phi(\theta,T_{\rm j}(\theta))=\frac{T_{\rm j}(\theta)(\theta + B)+D}{A T_{\rm j}(\theta) + \delta \theta+C} \]
is a $j$-integral of \eqref{ggen}. Using \eqref{uint} and its consequences, we can express $u_{k,0}$, $k=1,2,3$, in terms of $u$, $v$, $v_{1,0}$, $v_{2,0}$ and substitute these expressions into \eqref{gin}. This gives rise to the formula
\[ I[v]= \frac{(v_{2,0}+v_{1,0} + C - \delta B) (v_{1,0}+v + C - \delta B)}{v_{1,0}^2 + v_{1,0} (C - \delta B) - \delta A D} \]
for an $i$-integral of the equation \eqref{ggen}.

It should be noted that \eqref{ggen} coincides with the discrete Liouville equation \eqref{dle} if $C = -1$, $B \ne 0$, $A=\delta=0$. The equation \eqref{ggen} also takes the form $v_{1,1} (v+1) = v_{1,0} (v_{0,1} +1)$ (compare \eqref{ex1}) if $C=D=0$, $B=-1$, $A \ne -1$, $\delta =1$. In general, any Darboux integrable equation of the form 
\[ T_{\rm i}(\Omega (\tilde{v}, \tilde{v}_{0,1})) = \Psi(\tilde{v}, \tilde{v}_{0,1}),\qquad \frac{\partial \Omega}{\partial \tilde{v}} \frac{\partial \Psi}{\partial \tilde{v}_{0,1}} - \frac{\partial \Omega}{\partial \tilde{v}_{0,1}} \frac{\partial \Psi}{\partial \tilde{v}} \ne 0, \]
possessing a second-order autonomous $j$-integral is related to an equation of the form \eqref{ggen} via a point transformation $v=\mu(\tilde{v})$. This can be proved by reasonings that is similar to those used in \cite{foi}, but the rigorous proof of this proposition is beyond the scope of the present article.           

\section*{Acknowledgments} This work is partially supported by the Russian Foundation for Basic Research (grant number 13-01-00070-a).


\begin{thebibliography}{10}

\bibitem{R2001} Tremblay S, Grammatios B and Ramani A 2001 Integrable lattice equations and their growth properties {\it Phys. Lett.} A {\bf 278} 319--24. 

\bibitem{ABS} Adler V E, Bobenko A I and Suris Yu B 2009 Discrete nonlinear hyperbolic equations. Classification of integrable cases {\it Funkt. Analiz Prilozh.} {\bf 43} 3--21 (in Russian)

Adler V E, Bobenko A I and Suris Yu B 2009 {\it Funct. Anal. Appl.} {\bf 43} 3--17 (Engl. transl.)

\bibitem{LY} Levi R and Yamilov R I 2009 The generalized symmetry method for discrete equation {\it J. Phys. A: Math. Theor.} {\bf 42}(45) 454012 

\bibitem{Mikh} Mikhailov A V, Wang J P and Xenitidis P 2011 Recursion operators, conservation laws and integrability conditions for difference equations {\it Teor. Mat. Fyz.} {\bf 167} 23--49 (in Russian) 

Mikhailov A V, Wang J P and Xenitidis P 2011 {\it Theor. Math. Phys.} {\bf 167} 421--43 (Engl. transl.) 

\bibitem{LS} Levi D and Scimiterna C 2011 Linearizability of Nonlinear Equations on a Quad-Graph by a Point, Two Points and Generalized Hopf-Cole Transformations {\it SIGMA} {\bf 7} 079

\bibitem{Cal} Calogero F 1990 Why are certain nonlinear PDEs both widely applicable and integrable? {\it What is Integrability?} ed V E Zakharov (Springer: New York) pp~1--62

\bibitem{umj} Startsev S Ya 2011 Necessary conditions of Darboux integrability for differential-difference equations of a special kind {\it Ufimskii matem. zhurn.} {\bf 3}(1) 80--84 (in Russian)
  
Startsev S Ya 2011 {\it Ufa Math. J.} {\bf 3}(1) 78-82 (Engl. transl.) 

\bibitem{HZS} Habibullin I T, Zheltukhina N and Sakieva A 2010 On Darboux-integrable semi-discrete chains {\it J. Phys. A: Math. Theor.} {\bf 43} 434017

\bibitem{GY} Garifullin R N and Yamilov R I 2012 Generalized symmetry classification of discrete equations of a class depending on twelve parameters {\it J. Phys. A: Math. Theor.}  {\bf 45} 345205 

\bibitem{Hirota} Hirota R 1987 Discrete two-dimensional Toda molecule equation {\it J. Phys. Soc. Japan} {\bf 56} 4285--8  

\bibitem{AdS} Adler V E and Startsev S Ya 1999 Discrete analogues of the Liouville equation {\it Teor. Mat. Fiz.} {\bf 121}(2) 271--85 (in Russian)

Adler V E and Startsev S Ya 1999 {\it Theor. Math. Phys.} {\bf 121}(2) 1484--95 (Engl. transl.) 

\bibitem{gour} Goursat E 1898 {\it Le\c cons sur l'int\'egration des \'equations
aux d\'eriv\'ees partielles du second ordre \`a deux variables 
ind\'ependantes} Tome~2 (Paris: Hermann)

\bibitem{ZhSok} Zhiber A V and Sokolov V V 2001	Exactly integrable hyperbolic equations of Liouville type {\it Usp. mat. nauk.} {\bf 56}(1) 63--106 (in Russian)

Zhiber A V and Sokolov V V 2001 {\it Russ. Math. Surv.} {\bf 56}(1) 61--101 (Engl. transl.) 

\bibitem{Sit} Startsev S Ya 2010 On non-point invertible transformations of difference and differential-difference equations {\it SIGMA} {\bf 6} 092 

\bibitem{HZS2} Habibullin I, Zheltukhina N and Sakieva A 2011 Discretization of hyperbolic type Darboux integrable equations preserving integrability {\it J. Math. Phys.} {\bf 52} 093507 	

\bibitem{foi} Startsev S Ya 2012 Darboux integrable differential-difference equations admitting a first-order integral {\it Ufimskii matem. zhurn.} {\bf 4}(3) 161--76 (in Russian)

Startsev S Ya 2012 {\it Ufa Math. J.} {\bf 4}(3) P.~159--73 (Engl. transl.) 

\bibitem{avz} Zhiber A V 1994 Quasilinear hyperbolic equations with an infinite-dimensional symmetry algebra {\it Izv. RAN. Ser. Mat.} {\bf 58}(4) 33–-54 (in Russian)

Zhiber A V 1995 {\it Russian Academy of Sciences. Izvestiya Mathematics} {\bf 45}(1) 33–-54 (Engl. transl.) 

\bibitem{RH} Rasin O G and Hydon P E 2007 Symmetries of integrable difference equations on the quad-graph {\it Stud. Appl. Math.} {\bf 119} 253–-69

\end{thebibliography}
\end{document}